%% file: main.tex
\PassOptionsToPackage{unicode}{hyperref}
\PassOptionsToPackage{hyphens}{url}
\PassOptionsToPackage{dvipsnames,svgnames,x11names}{xcolor}
\documentclass[12pt]{article}

\usepackage{amsthm}
\usepackage{amsmath,amssymb}
\usepackage{iftex}
\ifPDFTeX
  \usepackage[T1]{fontenc}
  \usepackage[utf8]{inputenc}
  \usepackage{textcomp} 
\else 
  \usepackage{unicode-math}
  \defaultfontfeatures{Scale=MatchLowercase}
  \defaultfontfeatures[\rmfamily]{Ligatures=TeX,Scale=1}
\fi
\usepackage{lmodern}
\ifPDFTeX\else  
\fi
\IfFileExists{upquote.sty}{\usepackage{upquote}}{}
\makeatletter
\@ifundefined{KOMAClassName}{
  \IfFileExists{parskip.sty}{%
    \usepackage{parskip}
  }{
    \setlength{\parindent}{0pt}
    \setlength{\parskip}{6pt plus 2pt minus 1pt}}
}{
  \KOMAoptions{parskip=half}}
\makeatother
\usepackage{xcolor}
\makeatletter
\ifx\paragraph\undefined\else
  \let\oldparagraph\paragraph
  \renewcommand{\paragraph}{
    \@ifstar
      \xxxParagraphStar
      \xxxParagraphNoStar
  }
  \newcommand{\xxxParagraphStar}[1]{\oldparagraph*{#1}\mbox{}}
  \newcommand{\xxxParagraphNoStar}[1]{\oldparagraph{#1}\mbox{}}
\fi
\ifx\subparagraph\undefined\else
  \let\oldsubparagraph\subparagraph
  \renewcommand{\subparagraph}{
    \@ifstar
      \xxxSubParagraphStar
      \xxxSubParagraphNoStar
  }
  \newcommand{\xxxSubParagraphStar}[1]{\oldsubparagraph*{#1}\mbox{}}
  \newcommand{\xxxSubParagraphNoStar}[1]{\oldsubparagraph{#1}\mbox{}}
\fi
\makeatother

\usepackage{longtable,booktabs,array}
\usepackage{calc} 
\usepackage{etoolbox}
\makeatletter
\patchcmd\longtable{\par}{\if@noskipsec\mbox{}\fi\par}{}{}
\makeatother
\IfFileExists{footnotehyper.sty}{\usepackage{footnotehyper}}{\usepackage{footnote}}
\makesavenoteenv{longtable}
\usepackage{graphicx}
\makeatletter
\def\maxwidth{\ifdim\Gin@nat@width>\linewidth\linewidth\else\Gin@nat@width\fi}
\def\maxheight{\ifdim\Gin@nat@height>\textheight\textheight\else\Gin@nat@height\fi}
\makeatother
\setkeys{Gin}{width=\maxwidth,height=\maxheight,keepaspectratio}
\makeatletter
\def\fps@figure{htbp}
\makeatother

\makeatletter
\@ifpackageloaded{caption}{}{\usepackage{caption}}
\AtBeginDocument{%
\ifdefined\contentsname
  \renewcommand*\contentsname{Table of contents}
\else
  \newcommand\contentsname{Table of contents}
\fi
\ifdefined\listfigurename
  \renewcommand*\listfigurename{List of Figures}
\else
  \newcommand\listfigurename{List of Figures}
\fi
\ifdefined\listtablename
  \renewcommand*\listtablename{List of Tables}
\else
  \newcommand\listtablename{List of Tables}
\fi
\ifdefined\figurename
  \renewcommand*\figurename{Figure}
\else
  \newcommand\figurename{Figure}
\fi
\ifdefined\tablename
  \renewcommand*\tablename{Table}
\else
  \newcommand\tablename{Table}
\fi
}
\@ifpackageloaded{float}{}{\usepackage{float}}
\floatstyle{ruled}
\@ifundefined{c@chapter}{\newfloat{codelisting}{h}{lop}}{\newfloat{codelisting}{h}{lop}[chapter]}
\floatname{codelisting}{Listing}

\makeatother
\makeatletter
\@ifpackageloaded{caption}{}{\usepackage{caption}}
\@ifpackageloaded{subcaption}{}{\usepackage{subcaption}}
\makeatother

\ifLuaTeX
  \usepackage{selnolig}  
\fi
\usepackage[]{natbib}
\usepackage{bookmark}

\IfFileExists{xurl.sty}{\usepackage{xurl}}{} 
\hypersetup{
  pdftitle={Title},
  pdfauthor={Author 1; Author 2},
  pdfkeywords={3 to 6 keywords, that do not appear in the title},
  colorlinks=true,
  linkcolor={blue},
  filecolor={Maroon},
  citecolor={Blue},
  urlcolor={Blue},
  pdfcreator={LaTeX via pandoc}}

\newcommand{\anon}{1}

\input{notation}
\newtheorem{assumption}{Assumption}

\newtheorem{theorem}{Theorem}
\newtheorem{remark}{Remark}

\input{table/application_numbers.tex}

\begin{document}

\def\spacingset#1{\renewcommand{\baselinestretch}%
{#1}\small\normalsize} \spacingset{1}


\if1\anon
{
  \title{\bf Power and Sample Size Calculations for Hybrid Controlled Trials}
\author[1,2]{Ke Zhu}
\author[1]{Shu Yang}
\author[2]{Xiaofei Wang\footnote{Address for correspondence: Xiaofei Wang, Department of Biostatistics and Bioinformatics, Duke University, Durham, NC 27710, U.S.A. Email: xiaofei.wang@duke.edu}}
\affil[1]{\small Department of Statistics, North Carolina State University, Raleigh, NC 27695, U.S.A.}
\affil[2]{Department of Biostatistics and Bioinformatics, Duke University, Durham, NC 27710, U.S.A.}
  \date{}
  \maketitle
} \fi

\if0\anon
{
  \bigskip
  \bigskip
  \bigskip
  \begin{center}
    {\LARGE\bf Power and Sample Size Calculations for Hybrid Controlled Trials}
\end{center}
  \medskip
} \fi

\bigskip
\begin{abstract}
Hybrid controlled trials (HCTs) augment randomized controls with external controls (ECs) to address practical challenges in randomized controlled trials (RCTs) and improve statistical power in settings such as rare diseases, oncology, and pediatrics. However, prospective sample-size determination is challenging because the required RCT sample size depends on the comparability of ECs with RCT controls, which are unavailable at the planning stage. We propose a 5+3 design for HCT sample-size determination based on an inverse probability weighting estimator of the average treatment effect. The framework uses five conventional RCT design parameters and three additional scalar parameters characterizing EC comparability: the number of outcome-drift-free ECs, an overlap coefficient for the covariate distributions of the RCT and ECs, and a correlation coefficient linking the sampling mechanism to the control potential outcome. We establish the asymptotic distribution of the estimator and prove that its variance is determined by these design parameters under the proposed working models, yielding sample-size calculations for both continuous and binary outcomes. Simulation studies evaluate finite-sample performance, and a real clinical application illustrates {its} practical use. The method is implemented in the \texttt{hctdesign} R package.
\end{abstract}

\noindent%
{\it Keywords:} Causal inference; External control; Inverse probability weighting; Propensity score; Real-world evidence.
\vfill

\spacingset{1.8} 

\section{Introduction}

Randomized controlled trials (RCTs) remain the gold standard for evaluating new interventions, but a fully randomized comparison may be infeasible, inefficient, or ethically challenging. Such difficulties arise particularly in rare diseases, oncology, and pediatric settings, as well as {when} effective existing therapies or strong patient preferences limit willingness to enroll in a randomized control arm \citep{rahman2021leveraging,mishra2022external,ye2024considerations,selukar2025synthetic}. Hybrid controlled trials (HCTs) offer a potential solution by retaining randomization while augmenting the randomized control arm with external controls (ECs) drawn from historical trials or real-world data \citep{pocock1976combination,ventz2022design}. Appropriate borrowing of ECs can improve precision, reduce the number of randomized controls required, or permit unequal randomization in favor of the experimental treatment. A growing literature has developed both Bayesian and frequentist methods for analyzing HCTs \citep{zhu2026review}.

Beyond the choice of analysis method, a central design question for HCTs is how many new patients should be randomized, a decision that must be made before recruitment begins. This is important for at least two reasons. First, regulatory guidance emphasizes prespecification when ECs are incorporated into analyses to enhance credibility and guard against post hoc choices \citep{FDA2023}. Second, sponsors need to determine the RCT sample size before trial initiation, and this choice directly affects operating characteristics, recruitment burden, and cost. Despite the {growing} methodological literature on HCT analysis, methods for prospective sample-size determination remain comparatively limited.

A fundamental difficulty is that the required RCT sample size depends on the comparability of ECs with RCT controls, yet RCT controls are unavailable at the planning stage. Existing approaches address this difficulty in several ways. Bayesian methods characterize borrowing through the effective sample size induced by a prior or borrowing model \citep{zhang2022bayesian,bi2023beats,chen2024sequential,tian2025beam}. Frequentist approaches require specification of a sampling propensity score function that characterizes covariate shift between the RCT and ECs \citep{gao2025designing}, which can be difficult to elicit before trial enrollment. Other frequentist procedures use interim RCT data to adapt the amount of borrowing or reassess sample size \citep{guo2024adaptive,kojima2026sample}. Thus, there remains a need for a prospective sample-size method for HCTs that requires only a small number of interpretable design inputs, does not rely on access to interim data, and explicitly accounts for EC comparability.

We address this problem through a ``5+3'' parameterization for sample-size determination based on an inverse probability weighting (IPW) estimator of the average treatment effect. The first five inputs coincide with those commonly used in conventional RCT planning: the treatment effect, treatment-allocation proportion, type I error level, desired power, and an outcome-distribution parameter, such as the variance for a continuous outcome or the control response rate for a binary outcome. We then introduce three scalar parameters specific to HCTs: the number of outcome-drift-free ECs, an overlap coefficient describing the similarity between the covariate distributions of the RCT and ECs, and a correlation coefficient describing the association between the sampling mechanism and the control potential outcome. Under the proposed working models, these parameters determine the asymptotic variance of the IPW estimator and hence the required RCT sample size.
{We summarize the proposed 5+3 design framework in Figure~\ref{fig:diagram}, which provides a schematic overview of the required design inputs and their roles in sample-size determination.}

\begin{figure}[t]
    \centering
    \includegraphics[width=\linewidth]{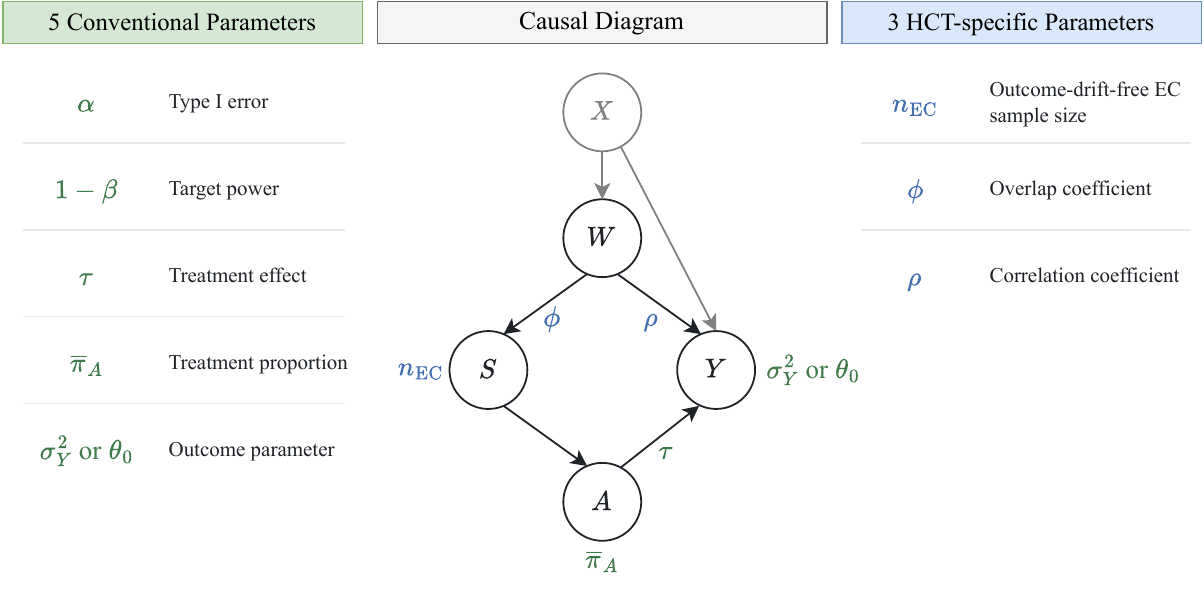}
    \caption{Schematic illustration of the proposed 5+3 design framework.}
    \label{fig:diagram}
\end{figure}

Our approach uses the \textit{linear sampling propensity score}, defined as the logit of the probability of entering the RCT under a logistic sampling model, as the core design quantity. The sampling propensity score is theoretically the coarsest balancing score, reducing multivariable covariate differences between the RCT and ECs to a single scalar while retaining the information needed for covariate balance. Analogous to the linear treatment propensity score of \citet{rubin1992characterizing}, we assume that its logit follows a normal distribution.
Our work is most closely related to the sample-size framework of \citet{liu2025sample} for observational studies, but differs in three important respects. {First, extending their framework to HCTs requires additional development because treatment effects are estimated by combining randomized and external data, introducing HCT-specific identification assumptions and a distinct weighting structure. Second, we refine the numerical calibration of the propensity score distribution. \citet{liu2025sample} use an auxiliary Beta distribution to obtain a Normal approximation for the linear propensity score, which may not exactly preserve the prespecified design parameters. We instead use their approximate values to initialize a direct Normal calibration, ensuring that the final distribution matches the prespecified marginal RCT proportion and overlap coefficient. Third, we derive separate sample-size formulas for continuous and binary outcomes, with the latter directly incorporating the Bernoulli mean--variance relationship rather than relying on a continuous-outcome approximation.} We compare related work on sample size calculation in Table~\ref{tab:literature-comparison}.

The main contributions are threefold. First, we formulate a prospective sample-size framework for HCTs in terms of five conventional RCT design parameters and three interpretable parameters characterizing EC comparability. Second, we establish the asymptotic distribution of the IPW estimator and show that its variance, for both continuous and binary outcomes, is determined by these design parameters under the proposed models. Third, we evaluate the finite-sample performance of the proposed method through simulations and illustrate its use in a real clinical application. Reproducible code is available at \url{github.com/ke-zhu/hct-sample-size}, and an implementation is provided in the \texttt{hctdesign} R package at \url{github.com/ke-zhu/hctdesign}.

\begin{table}[t]
\centering
\caption{Comparison of existing design and sample-size approaches related to HCTs.}
\label{tab:literature-comparison}
\small
\setlength{\tabcolsep}{5pt}
\renewcommand{\arraystretch}{1.15}

\begin{tabularx}{\textwidth}{
>{\raggedright\arraybackslash}p{3.4cm}
>{\raggedright\arraybackslash}p{3.4cm}
>{\raggedright\arraybackslash}p{2.5cm}
>{\raggedright\arraybackslash}X
}
\toprule
Literature & Setting & Design stage & Required inputs \\
\midrule

\citet{zhang2022bayesian};
\citet{bi2023beats};
\citet{chen2024sequential};
\citet{tian2025beam}
&
HCT
&
Prospective / adaptive
&
Prior/borrowing model / interim RCT data
\\
\addlinespace[3pt]

\citet{guo2024adaptive};
\citet{kojima2026sample}
&
HCT
&
Adaptive
&
Interim RCT data
\\
\addlinespace[3pt]

\citet{gao2025designing}
&
HCT; SAT+EC
&
Prospective
&
Nuisance functions
\\
\addlinespace[3pt]

\citet{liu2025sample}
&
Observational study
&
Prospective
&
Scalar parameters
\\

\midrule

Proposed method
&
HCT
&
Prospective
&
5 conventional + 3 HCT-specific scalar parameters
\\

\bottomrule
\end{tabularx}
\end{table}

\section{Hybrid Controlled Trials}\label{sec:pre}

Consider an HCT with total sample size $n$, consisting of an RCT with indicator $S=1$ and sample size $\nrct$, and an EC dataset with indicator $S=0$ and sample size $\nec$. The RCT population is regarded as the target population. Let $Y(1)$ and $Y(0)$ denote the potential outcomes under treatment $a=1$ and control $a=0$, respectively. The estimand of interest is the average treatment effect (ATE) in the RCT population,
$$
\tau=\theta_1-\theta_0,
\quad
\theta_a=\mathbb{E}\{Y(a)\mid S=1\}.
$$

Within the RCT, treatment is randomized, with $n_1$ patients assigned to $A=1$ and $n_0$ patients assigned to $A=0$, and a known treatment allocation proportion $\bpA=n_1/\nrct$.
In the EC dataset, all patients receive control, so $A=0$ when $S=0$. 
Let $\pS(x)=\Pr(S=1\mid X=x)$ denote the sampling propensity score \citep{tipton2013improving}, which is typically unknown in practice and needs to be estimated. Let $\bar{\pi}_S=\Pr(S=1)=\nrct/n$ denote the RCT proportion.
Let $Y$ denote the observed outcome. 
Throughout, a subscript $i$ denotes a quantity evaluated for unit $i$, while the corresponding notation without $i$ denotes its generic population counterpart.
We impose the following identifiability assumptions.

\begin{assumption}[Identifiability assumptions for RCT]
\label{ass:rct}
For units with $S=1$, the following hold: (i) Stable Unit Treatment Value Assumption (SUTVA) \citep{Rubin1980}: $Y=Y(A)$. (ii) Unconfoundedness of treatment: $\{Y(0),Y(1)\}\perp A \mid X,S=1$. (iii) Positivity of treatment: $\bpA\in(0,1)$ is known by design.
\end{assumption}

\begin{assumption}[SUTVA for EC]
For units with $S=0$, $Y=Y(0)$.
\end{assumption}

\begin{assumption}[Conditional mean exchangeability]
\label{ass:ec}
$
\mathbb{E}\{Y(0)\mid S=1,X\}=\mathbb{E}\{Y(0)\mid S=0,X\}.
$
\end{assumption}

\begin{assumption}[Positivity of sampling]\label{ass:pi}
$\pS(x)>0$ for all $x$ such that $f_{X}(x) > 0$.
\end{assumption}

Assumption~\ref{ass:rct} is typically guaranteed by the design of the RCT. Assumption~\ref{ass:ec} rules out \textit{outcome drift} in the ECs and could be relaxed through selective borrowing to allow partial violations \citep{gao2025improving,zhu2025enhancing,liu2025robust}. Assumption~\ref{ass:pi} can be satisfied by pre-restricting ECs to individuals who meet the eligibility criteria of the RCT.

Under these assumptions, we consider the inverse probability weighting (IPW) estimator \citep{Valancius2024}, which addresses the \textit{covariate shift} between RCT and EC:
$$
\hat{\tau}=\hat{\theta}_1-\hat{\theta}_0,\quad \hat{\theta}_a
=\frac{\sum_{i=1}^n w_{a,i}Y_i}{\sum_{i=1}^n w_{a,i}},\quad
w_{1,i}=\frac{S_iA_i}{\bpA},\quad w_{0,i}
=\frac{\pi_S(X_i)(1-A_i)}{1-\bpA\pi_S(X_i)}.
$$
For the control group,
$1-\bpA\pi_S(X)=\Pr(A=0\mid X)$ is the probability of being observed as a control, whereas the numerator $\pi_S(X)=\Pr(S=1\mid X)$ targets the RCT population. Thus, $w_0$ reweights the pooled RCT and ECs toward the RCT population, while $w_1$ accounts for treatment randomization within the RCT. 
{Notably, because of the presence of RCT controls, \(w_0\leq (1-\bpA)^{-1}\), with equality when \(\pi_S(X)=1\). Therefore, \(w_0\) is inherently bounded and cannot produce arbitrarily extreme weights.}
The normalization by the sum of weights yields a Hájek-type IPW estimator \citep{khan2023adaptive}.
In this paper, we focus on the design stage. For a practical workflow for the analysis of HCTs, we refer readers to \cite{zhu2026robust}.

\begin{remark}[IPW estimator with estimated sampling propensity score]
{At the design stage, we conceptually use the IPW estimator with the oracle sampling propensity score to obtain a tractable sample-size calculation. At the analysis stage, however, $\pS(\cdot)$ would be estimated from the pooled RCT and EC data. Using an estimated rather than the true propensity score can reduce the asymptotic variance of IPW estimators, but the corresponding variance expression involves additional nuisance quantities that are more difficult to specify at the design stage \citep{lunceford2004stratification}. If one wishes to account for propensity score estimation at the design stage, the variance can be refined through the corresponding influence-function correction or evaluated by simulation under the working models. In our simulations, the empirical power using an estimated sampling propensity score remains close to the theoretical power based on the oracle score; see Figure~\ref{fig:simulation-power-rho03}.}
\end{remark}

\begin{remark}[Efficient estimators]
Semiparametric efficient estimators may be used at the analysis stage, including RCT-only covariate adjustment \citep[e.g.,][]{bannick2025general,liu2025coadvise}, prognostic adjustment \citep{schuler2022increasing}, and EC-borrowing estimators such as augmented IPW (AIPW) \citep{li2023improving}, targeted maximum likelihood estimation (TMLE) \citep{Valancius2024}, and augmented calibration weighting (ACW) \citep{gao2025improving}. These approaches may yield greater power than that projected by the IPW-based calculation. However, characterizing their asymptotic power requires additional information about the outcome model \citep{schuler2022designing,gao2025designing}. By focusing on the IPW estimator, our framework provides a complementary approach for information-limited settings with simpler sample-size calculations. This reflects a trade-off among information requirements, simplicity, and statistical efficiency, and practitioners may choose the approach that best matches the information available at the design stage.
\end{remark}

\section{Power and Sample Size Calculations}

\subsection{Power Formula}

We first establish the asymptotic distribution of the IPW estimator for HCTs.

\begin{theorem}[Asymptotic normality of the IPW estimator]
\label{thm:ipw_asymp}
Suppose Assumptions~\ref{ass:rct}--\ref{ass:pi} hold. In addition, assume that
(i) $\mathbb{E}\{Y(a)^2\mid S=1\}<\infty$ for $a=0,1$; and
(ii) for the control potential outcome, the conditional second moment is transportable across the RCT and EC samples, in the sense that $\mathbb{E}\left[
\{Y(0)-\theta_0\}^2\mid X,S=1
\right]=
\mathbb{E}\left[
\{Y(0)-\theta_0\}^2\mid X,S=0
\right]$.
Then
\[
\sqrt{n}(\hat{\tau}-\tau)
\overset{d}{\longrightarrow}
N(0,V),
\qquad
V=V_1+V_0,
\]
where
\begin{equation}
\label{eq:var}
V_1
=
\frac{\mathbb{V}\{Y(1)\mid S=1\}}
{\bpA\bar{\pi}_S},
\quad
V_0
=
\frac{1}{\bar{\pi}_S^2}
\mathbb{E}\left[
\frac{\pS(X)^2}{1-\bpA\pS(X)}
\{Y(0)-\theta_0\}^2
\right].
\end{equation}
\end{theorem}
\begin{remark}
Condition (i) is a standard moment condition for asymptotic normality and finite variance. {Condition (ii) is introduced to obtain a parsimonious expression for $V_0$ that can be parameterized using quantities that are interpretable and practically specifiable at the design stage. Although ECs may be more heterogeneous than RCT controls in practice, Condition (ii) is imposed only on the outcome-drift-free ECs and concerns the conditional second moment given covariates, rather than the marginal outcome variance. Moreover, Condition (ii) is not required for identification or consistency of the IPW estimator. If it is relaxed, the same general variance calculation remains applicable, but $V_0$ depends on additional data-source-specific conditional second-moment quantities, which must then be specified as additional design parameters.}
\end{remark}

\begin{remark}[Special cases]
The proposed framework includes two important special cases. When $\bpA=1$, all trial participants are assigned to treatment, corresponding to a single-arm trial (SAT) with ECs. In this case, $\tau$ is the average treatment effect on the treated (ATT), and our framework reduces to the ATT sample-size calculation of \citet{liu2025sample}. When $\pS(X)\equiv 1$, no ECs are included, and our framework reduces to the conventional RCT sample-size calculation.
\end{remark}

Based on Theorem~\ref{thm:ipw_asymp}, Wald-type inference can be performed with a consistent variance estimator. For the one-sided test $H_0:\tau\leq 0$ versus $H_1:\tau>0$ at significance level $\alpha$, the asymptotic power is
\begin{equation}
\label{eq:power}
1-\beta
=
\Phi\left(
\sqrt{\frac{\nrct+\nec}{V}}\,\tau
-
z_{1-\alpha}
\right),
\end{equation}
where $z_{1-\alpha}$ denotes the $(1-\alpha)$ quantile of the standard normal distribution. The required RCT sample size $\nrct$ can be obtained numerically by solving \eqref{eq:power}.

\begin{remark}[Asymptotic approximation in small trials]
{The power formula \eqref{eq:power} relies on an asymptotic approximation, whereas HCTs are particularly attractive in rare-disease and pediatric settings where $\nrct$ may be small. With few randomized patients, the normal approximation and plug-in variance estimation may be inaccurate, and the Wald test may fail to maintain the nominal type I error. We therefore recommend evaluating the operating characteristics of any candidate design by simulation under the working models before finalizing $\nrct$. Our simulations consider $\nrct\geq100$; for smaller trials, safeguards such as $t$-distribution critical values or a modest variance inflation factor may improve finite-sample performance. Alternatively, Fisher randomization test can provide finite-sample exact type I error control under the sharp null in HCTs \citep{zhu2025enhancing,liu2025robust}.}
\end{remark}

The key challenge is to determine $V$ at the design stage, when only limited information may be available. For a conventional RCT, sample size calculations are typically characterized by \textit{five parameters}: the treatment effect $\tau$, treatment allocation proportion $\bpA$, type I error level $\alpha$, desired power $1-\beta$, and an outcome parameter, such as the variance $\sigma_Y^2$ for a continuous outcome or the control response rate $\theta_0$ for a binary outcome.
However, HCTs require additional information because the precision gained from ECs depends not only on the number of ECs, but also on their comparability with the RCT population. 

In the following subsections, we introduce \textit{three additional scalar parameters} that summarize these HCT-specific features and, together with the five standard parameters, are sufficient to determine $V$ for power and sample size calculations. We refer to this framework as the \textit{5+3 design}.

\subsection{Outcome-Drift-Free EC Sample Size}

We have so far used $\nec$ to denote the EC sample size. Theorem~\ref{thm:ipw_asymp}, however, requires the ECs contributing to the IPW estimator to satisfy the conditional mean exchangeability condition in Assumption~\ref{ass:ec}. We therefore interpret $\nec$ more specifically as the number of \textit{outcome-drift-free} ECs, that is, ECs for which no residual outcome drift remains after conditioning on $X$. Let $\tilde n_{\rm EC}$ denote the total number of ECs available prior to accounting for possible outcome drift. In general, we have $\nec \leq \tilde n_{\rm EC}$, with equality when all available ECs are assumed to satisfy Assumption~\ref{ass:ec}.

At the design stage, when no RCT control data are available to directly assess potential outcome drift, $\nec$ may be specified by discounting the available EC sample size,
\[
\nec=c\,\tilde n_{\rm EC},
\qquad c\in[0,1],
\]
where $c$ represents the anticipated proportion of ECs that satisfy the outcome-drift-free working assumption. Sample size calculations may be conducted over a range of plausible values of $c$ to assess sensitivity to this specification.

\begin{remark}[Choice of $\nec$ and residual outcome drift]
{Because only limited information and no concurrent RCT control data are available at the design stage, the proposed framework necessarily relies on working assumptions about the anticipated number of outcome-drift-free ECs and their satisfaction of Assumption~\ref{ass:ec}. Consequently, misspecification of $c$ or residual outcome drift among ECs assumed to be outcome-drift-free may lead to biased estimation, inflated type I error, and inaccurate power calculations. Several strategies can mitigate these limitations: (i) when pilot or interim RCT control data become available, a prespecified selective-borrowing procedure may be used to identify outcome-drift-free ECs \citep{gao2025improving,zhu2025enhancing,liu2025robust}, with selection consistency providing asymptotic protection by excluding ECs with outcome drift; (ii) design-stage simulations can introduce plausible levels of outcome drift to evaluate type I error and power and identify tipping points at which the operating characteristics deteriorate; and (iii) the required $\nrct$ can be evaluated over a range of $c$ values, including $c=0$ as the conventional-RCT benchmark, to assess sensitivity to the assumed proportion of outcome-drift-free ECs.}
\end{remark}

\subsection{Overlap Coefficient}

The second HCT-specific parameter characterizes covariate overlap between the RCT and EC populations. This overlap is reflected by the sampling propensity score $\pS(X)$, but its distribution is generally difficult to specify at the design stage. We therefore use a scalar overlap coefficient to summarize the information needed for sample size calculation.

Specifically, let $f_s(u)$ denote the density of $\pS(X)$ conditional on $S=s$, for $s=0,1$, corresponding to the sampling propensity score distributions in the EC and RCT populations, respectively. We adopt the Bhattacharyya coefficient \citep{bhattacharyya1943measure} as a scalar measure of overlap between $f_0(u)$ and $f_1(u)$:
\[
\phi
=
\int_0^1 \{f_0(u)f_1(u)\}^{1/2}\,\mathrm{d}u,
\qquad 0\leq\phi\leq1.
\]
We refer to $\phi$ as the \textit{overlap coefficient}.
Larger values of $\phi$ indicate greater overlap, with $\phi=1$ corresponding to identical distributions.

To translate this scalar measure of overlap into the distribution of $\pS(X)$, we impose a logistic model for the sampling propensity score. Specifically, we consider
\[
\logit\{\pS(X)\}=\widetilde X^{\T}\gamma,
\qquad
\widetilde X^{\T}=(1,X^{\T}),
\]
where $\logit(u)=\log\{u/(1-u)\}$ and $\expit(v)=\{1+\exp(-v)\}^{-1}$ denote the logit and inverse-logit functions, respectively. Analogous to the linear treatment propensity score \citep{rubin1992characterizing}, we define $W=\widetilde X^{\T}\gamma$ as the \textit{linear sampling propensity score} and assume $W\sim N(\mu_W,\sigma_W^2)$.

By Bayes' rule and $\pS(X)=\expit(W)$, we have
\[
\phi
=
\frac{
\mathbb{E}\!\left[
\{\pS(X)(1-\pS(X))\}^{1/2}
\right]
}{
\{\bar{\pi}_S(1-\bar{\pi}_S)\}^{1/2}
}
=
\frac{
\mathbb{E}\!\left[
\{\expit(W)(1-\expit(W))\}^{1/2}
\right]
}{
\{\bar{\pi}_S(1-\bar{\pi}_S)\}^{1/2}
}.
\]
Together with $\bar{\pi}_S=\mathbb{E}\{\pS(X)\}=\mathbb{E}\{\expit(W)\}$, these two equations determine $(\mu_W,\sigma_W^2)$ from $(\bar{\pi}_S,\phi)$ and can be solved numerically. 
In our numerical implementation, we first use the Beta-distribution-based approximation to the logit-normal sampling propensity score distribution in \citet{liu2025sample} to obtain initial values for $(\mu_W,\sigma_W^2)$. We then numerically refine these parameters to match the target $(\bar{\pi}_S,\phi)$ directly. This refinement retains the approximation as an efficient starting point while improving agreement between the resulting distribution of $W$ and the prespecified design parameters.

When solving \eqref{eq:power} for the required RCT sample size, each candidate $\nrct$ determines $\bar{\pi}_S=\nrct/(\nrct+\nec)$. Given a prespecified overlap coefficient $\phi$, we can then obtain $(\mu_W,\sigma_W^2)$ and hence the distribution of $\pS(X)$ required to evaluate $V_0$. Thus, $\phi$ provides a scalar summary of the covariate overlap needed for sample size calculation, without requiring specification of the full multivariate covariate distributions in the RCT and EC populations. Figure~\ref{fig:phi-overlap} illustrates the interpretation of $\phi$ through the similarity of the sampling propensity score distributions in the RCT and EC populations.

\begin{figure}[t]
    \centering
    \includegraphics[width=1\linewidth]{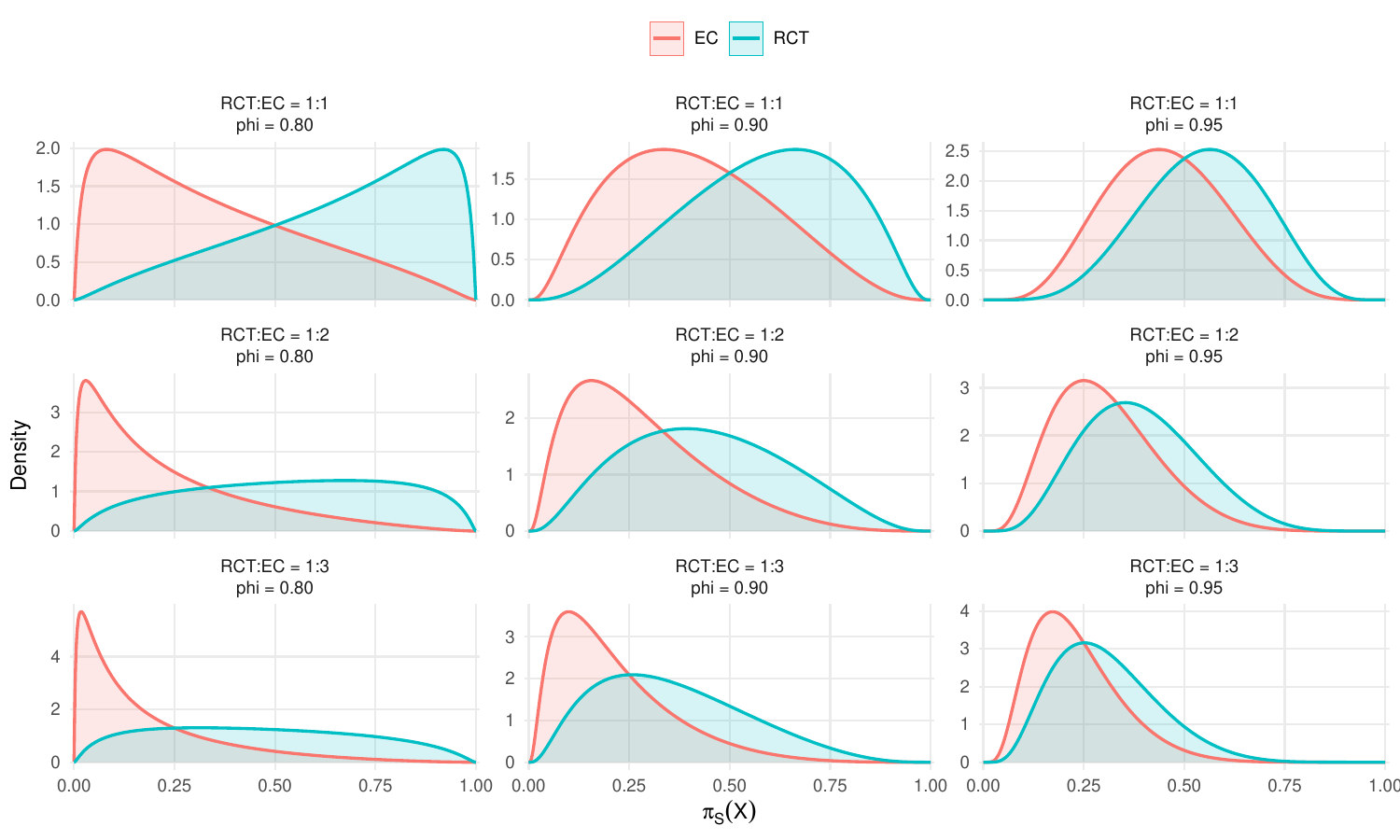}
    \caption{Illustration of the overlap coefficient $\phi$. Each panel shows the conditional distributions of the sampling propensity score $\pi_S(X)$ among RCT participants and ECs. When $\phi=1$, the two distributions coincide and degenerate to a common point mass at $\bpS$.}
    \label{fig:phi-overlap}
\end{figure}

\begin{remark}[Specification of $\phi$ and normality of $W$]
{Two limitations merit consideration. First, $\phi$ is specified at the design stage before RCT covariates are observed, and the required sample size can be sensitive to this choice. In the application of Section~\ref{sec:rd}, the operational RCT size increases from 140 to 220 as $\phi$ decreases from 1 (perfect overlap) to 0.80 (poor overlap), so an optimistic overlap assumption may leave the trial underpowered. Second, mapping $(\bpS,\phi)$ to the distribution of $\pS(X)$ relies on the normal working model for $W$, although our simulations suggest robustness to the mild nonnormality induced by mixed continuous and binary covariates. When individual-level EC covariates and information on the expected RCT baseline distribution are available, they can be used to inform both $\phi$ and a more realistic sampling propensity score distribution. Otherwise, we recommend evaluating sample size over a range of $\phi$ values and, when feasible, conducting a blinded reassessment after early RCT enrollment.}
\end{remark}

\subsection{Correlation Coefficient}

The third HCT-specific parameter characterizes the association between
the covariates and the control potential outcome that is relevant for
evaluating $V_0$. Because the IPW weight depends on $X$ only through
$\pS(X)=\expit(W)$, we summarize this association through the
one-dimensional linear sampling propensity score $W$.

Specifically, we define the \textit{correlation coefficient}
\[
\rho
=
\operatorname{cor}\{W,Y(0)\},
\qquad
-1\leq\rho\leq1.
\]
The parameter $\rho$ captures the correlation between the covariate direction that differentiates the RCT and EC populations and the control potential outcome. Thus, it depends not only on the prognostic strength of the covariates, but also on how closely the prognostic direction aligns with the covariate differences between the two populations. In the limiting case of perfect overlap ($\phi=1$), $W$ degenerates to a constant, and we set $\rho=0$ by convention.

\begin{remark}[Specification of $\rho$]
Among the three HCT-specific inputs, $\rho$ may be less familiar because it characterizes the correlation between the linear sampling propensity score $W$ and the control potential outcome. {For binary outcomes, $\rho$ has a universal upper bound $|\rho|\le\sqrt{2/\pi}\approx0.798$ because a binary variable cannot achieve perfect linear correlation with a normally distributed continuous variable \citep{gradstein1986maximal}.} For continuous outcomes, the magnitude of $\rho$ is instead bounded by the prognostic strength of the covariates, as $W$ is a linear function of $X$; thus, information from existing prognostic models, such as the multiple correlation coefficient or $R^2$, can inform a plausible upper bound for $\rho^2$. In practice, we recommend evaluating the required sample size over a range of plausible values of $\rho$.
\end{remark}

We next show how $\rho$, together with the conventional outcome distribution parameter, determines $V_0$ for binary and continuous outcomes.

\subsection{Variance Calculation for Binary Outcome}

For a binary outcome, $\theta_0=\mathbb{E}\{Y(0)\mid S=1\}$
denotes the control response rate in the RCT population. We posit the logistic model
\begin{equation}
\label{eq:Ybin}
\logit\left[
\mathbb{E}\{Y(0)\mid W\}
\right]
=
a+bW .
\end{equation}
The following result expresses $V_0$ in terms of the design parameters.
\begin{theorem}[Control-arm variance for binary outcomes]
\label{thm:binary_variance}
Suppose \eqref{eq:Ybin} holds. Then
\begin{equation}
\label{eq:V0_binary}
V_0
=
\frac{1}{\bar{\pi}_S^2}
\mathbb{E}\left[
\left\{
(1-2\theta_0)\expit(a+bW)+\theta_0^2
\right\}
\frac{\expit(W)^2}
{1-\bpA\expit(W)}
\right],
\end{equation}
where $(a,b)$ are determined by the design parameters $(\theta_0,\rho)$ through
\begin{equation}
\label{eq:binary_mean}
\theta_0
=
\frac{1}{\bar{\pi}_S}
\mathbb{E}\left[
\expit(W)\expit(a+bW)
\right],
\end{equation}
and
\begin{equation}
\label{eq:binary_corr}
\rho
=
\frac{
\mathbb{E}\left[
(W-\mu_W)\expit(a+bW)
\right]
}{
\sigma_W
\left[
\mathbb{E}\{\expit(a+bW)\}
\{1-\mathbb{E}\{\expit(a+bW)\}\}
\right]^{1/2}
}.
\end{equation}
\end{theorem}
\begin{remark}[Perfect overlap]
When $\phi=1$, the sampling propensity score is constant,
$\pS(X)\equiv\bar{\pi}_S$, so that $W$ is degenerate and
$\sigma_W=0$. In this case, we set $\rho=0$ by convention,
$\mathbb{E}\{Y(0)\mid W\}=\theta_0$, and
\[
V_0
=
\frac{\theta_0(1-\theta_0)}
{1-\bpA\bar{\pi}_S}.
\]
\end{remark}

The expectation in \eqref{eq:V0_binary} is taken over
$W\sim N(\mu_W,\sigma_W^2)$ and can be evaluated numerically.
For the treatment arm,
\[
V_1
=
\frac{\theta_1(1-\theta_1)}
{\bpA\bar{\pi}_S}.
\]
Therefore, for binary outcomes, $V=V_1+V_0$ is determined by the conventional design parameters $(\tau,\theta_0,\bpA)$, the three HCT-specific parameters $(\nec,\phi,\rho)$, and the candidate RCT sample size $\nrct$.

\subsection{Variance Calculation for Continuous Outcome}

For a continuous outcome, $\sigma_Y^2=\mathbb{V}\{Y(0)\mid S=1\}$ 
denotes the control outcome variance in the RCT population. We posit the homoscedastic linear model
\begin{equation}
\label{eq:Ycon}
Y(0)=a+bW+\varepsilon,
\qquad
\mathbb{E}(\varepsilon\mid W,S)=0,
\qquad
\mathbb{V}(\varepsilon\mid W,S)=\sigma_\varepsilon^2,
\end{equation}
which satisfies the conditional second-moment transportability condition in Theorem~\ref{thm:ipw_asymp}.
The following result expresses $V_0$ in terms of the design parameters.
\begin{theorem}[Control-arm variance for continuous outcomes]
\label{thm:continuous_variance}
Suppose \eqref{eq:Ycon} holds. Then
\begin{equation}
\label{eq:V0_continuous}
V_0
=
\frac{\sigma_Y^2\kappa}
{\bar{\pi}_S^2}
\mathbb{E}\left[
\left\{
(1-\rho^2)
+
\rho^2
\frac{(W-\mu_{W,1})^2}{\sigma_W^2}
\right\}
\frac{\expit(W)^2}
{1-\bpA\expit(W)}
\right],
\end{equation}
where 
\[
\mu_{W,1}=\bar{\pi}_S^{-1}\mathbb{E}\{\expit(W)W\},\quad
\sigma_{W,1}^2=\bar{\pi}_S^{-1}\mathbb{E}\{\expit(W)(W-\mu_{W,1})^2\},
\]
and
\[
\kappa
=
\frac{\sigma_W^2}
{\rho^2\sigma_{W,1}^2+(1-\rho^2)\sigma_W^2}.
\]
\end{theorem}
\begin{remark}[Interpretation of $\kappa$]
The parameter $\kappa$ captures the additional complexity induced by covariate shift in the linear sampling propensity score. Notably, $\kappa$ is a derived quantity rather than an additional HCT-specific parameter.
\end{remark}
\begin{remark}[Perfect overlap]
When $\phi=1$, the sampling propensity score is constant,
$\pS(X)\equiv\bar{\pi}_S$, so that $W$ is degenerate and
$\sigma_W=0$. In this case, we set $\rho=0$ and $\kappa=1$ by convention and evaluate
the control-arm variance directly as
\[
V_0
=
\frac{\sigma_Y^2}{1-\bpA\bar{\pi}_S}.
\]
\end{remark}
Under the conventional equal-variance assumption, $\mathbb{V}\{Y(1)\mid S=1\}=\mathbb{V}\{Y(0)\mid S=1\}=\sigma_Y^2$, the treatment-arm variance is
\[
V_1=\frac{\sigma_Y^2}{\bpA\bar{\pi}_S}.
\]
Therefore, for continuous outcomes, $V=V_1+V_0$ is determined by the conventional design parameters $(\tau,\sigma_Y^2,\bpA)$, the three HCT-specific parameters $(\nec,\phi,\rho)$, and the candidate RCT sample size $\nrct$.

\section{Simulation Studies}
\subsection{Simulation Design}

We assessed the finite-sample performance of the proposed power calculation using a fixed RCT target population and independently generated, non-nested EC populations. This design held the estimand, the ATE in the RCT population, fixed while varying only the EC population across overlap settings, enabling meaningful power comparisons across scenarios.

We generated five baseline covariates. The first three followed a multivariate normal distribution with mean zero, unit variances, and $\operatorname{cor}(X_j,X_k)=0.2^{|j-k|}$; the remaining two were independent Bernoulli covariates with prevalences 0.5 and 0.3. This defined the RCT target population. The EC covariate distribution was generated by exponential tilting, $f_{\rm EC}(x;\boldsymbol{\eta})=f_{\rm RCT}(x)\exp(-\boldsymbol{\eta}^{\T}x)/M(\boldsymbol{\eta})$, where $M(\boldsymbol{\eta})=\mathbb{E}_{\rm RCT}\{\exp(-\boldsymbol{\eta}^{\T}X)\}$. The resulting true linear sampling propensity score was $W=\logit(\bar\pi_S)+\log M(\boldsymbol{\eta})+\boldsymbol{\eta}^{\T}X$. For each setting, the magnitude and direction of $\boldsymbol{\eta}$ were calibrated to the target $\phi$ and $\rho$, respectively. Because this data-generating mechanism does not require $W$ to be exactly normal, it also evaluates robustness to mild model misspecification.

For continuous outcomes, we generated $Y(a)=a\tau+\boldsymbol{\beta}^{\T}\left\{X-(0,0,0,0.5,0.3)^{\T}\right\}+\varepsilon_a$, where $\varepsilon_a\sim N(0,0.6)$, with $0.6$ denoting the variance. We set $\boldsymbol{\beta}\propto(0.5,0.4,0.3,0.5,0.4)^{\T}$ and scaled it so that the prognostic component explained 40\% of the RCT outcome variance, yielding $\sigma_Y^2=1$. We set $\tau=0.30$. For binary outcomes, logistic models used the same covariate coefficients in both arms, with intercepts calibrated to RCT marginal response rates of 0.30 and 0.46, giving $\tau=0.16$. In both settings, EC outcomes followed the same conditional control outcome model as the RCT and were therefore outcome-drift-free by construction.

We considered $\nrct\in\{100,200,500\}$, fixed $\nec=300$, and assigned three-quarters of the RCT participants to treatment. The overlap/correlation settings
were $(\phi,\rho)=(1,0)$ and all combinations of
$\phi\in\{0.95,0.90,0.80\}$ and $\rho\in\{0.1,0.3\}$. Within each replicate and
outcome-by-$\nrct$ setting, the same realized RCT data were reused across all
seven overlap/correlation settings, whereas a new independent EC sample was
generated for each setting. Each generated dataset was analyzed using both the
true sampling propensity score and a correctly specified estimated logistic
sampling propensity score. We evaluated a one-sided test at level 0.025 using
10,000 Monte Carlo replicates per setting.

\begin{figure}[t]
\centering
\includegraphics[width=\linewidth]{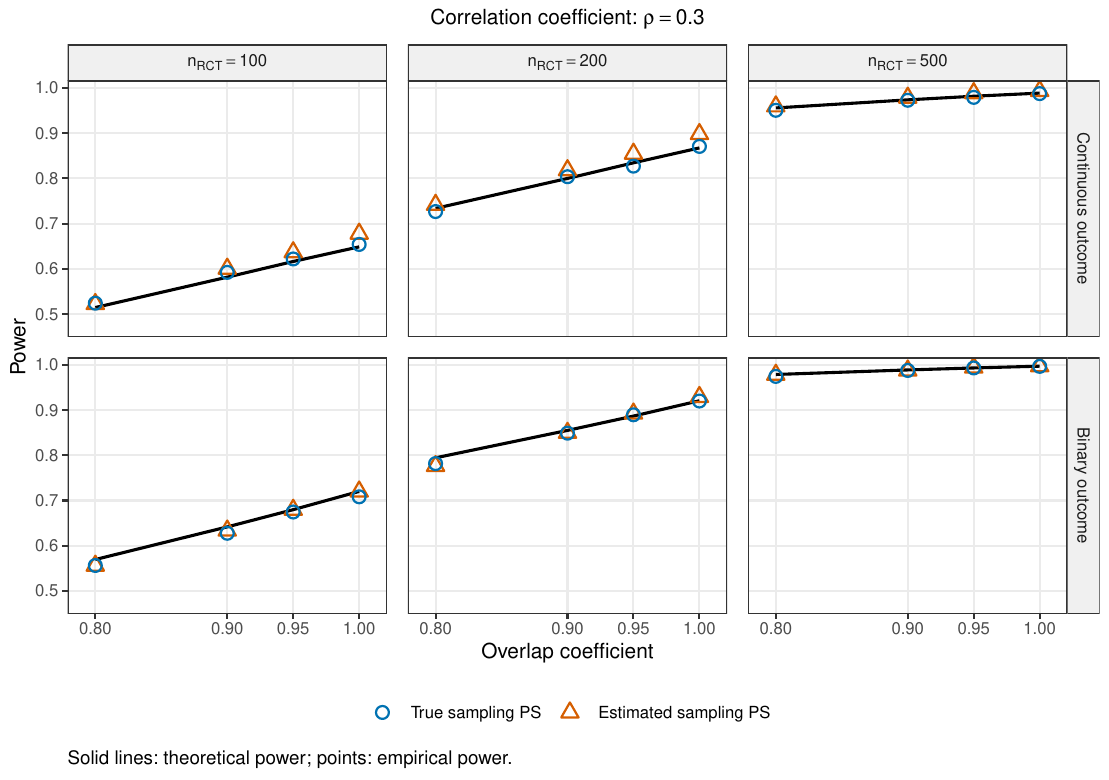}
\caption{Theoretical and empirical power for correlation coefficient $\rho=0.3$, by outcome type, RCT sample size, and overlap coefficient. Solid lines show the proposed theoretical power, and points show empirical power from 10,000 replicates. The perfect-overlap endpoint has $\rho=0$ by convention.}
\label{fig:simulation-power-rho03}
\end{figure}

\subsection{Simulation Results}

Figure~\ref{fig:simulation-power-rho03} compares empirical and theoretical power for $\rho=0.3$; the corresponding results for $\rho=0.1$, shown in Appendix Figure~\ref{fig:simulation-power-rho01}, display the same pattern. 
The theoretical power, calculated based on the true sampling propensity score, closely matched the corresponding empirical power with the true score, with a maximum absolute difference of 0.014 across all settings. Empirical power with the estimated score was slightly higher, with a maximum absolute deviation of 0.031 from the theoretical power.
Overall, empirical power increased with the overlap coefficient. At a fixed imperfect-overlap level, larger $\rho$ generally reduced power. Appendix Table~\ref{tab:simulation-power} reports the detailed power numbers for both binary and continuous outcomes. Detailed bias, empirical standard deviation, mean estimated standard error, and coverage results are reported in Appendix Tables~\ref{tab:simulation-validation-continuous} and~\ref{tab:simulation-validation-binary}.

\section{Application}\label{sec:rd}

\input{table/application_parameters.tex}

We illustrate the proposed calculation using the INJeCT-LUNG randomized phase II trial in resectable stage IB--IIIB non-small-cell lung cancer. The experimental arm adds intratumoral large-surface-area microparticle paclitaxel to neoadjuvant chemotherapy and checkpoint inhibition, whereas the control arm receives neoadjuvant chemotherapy and checkpoint inhibition alone. The protocol specifies 48-month event-free survival (EFS) as the primary outcome. The planned analysis uses survival pseudo-values to account for censoring \citep{andersen2010pseudo} and may augment the concurrent controls with ECs from CheckMate 816 \citep{forde2025overall} and CheckMate 77T \citep{cascone2024perioperative}.

For the design-stage calculation, we use the protocol-specified marginal 48-month EFS rates and define a binary endpoint based on EFS status at the 48-month landmark. Table~\ref{tab:application-parameters} summarizes the five conventional and three HCT-specific inputs together with their practical interpretations for communication with the clinical team. The five conventional inputs are taken directly from the protocol: the two EFS rates, 3:1 treatment allocation, a one-sided type I error of 0.10 for phase II screening, and 80\% target power. The three HCT-specific inputs are set to $\nec=300$, $\phi=0.90$, and $\rho=0.20$. Here, $\nec=300$ assumes that all 300 candidate EC patients are retained as outcome-drift-free after the prespecified compatibility assessment. The overlap coefficient $\phi$ is reported on a 0--1 scale: $\phi=1$ represents perfect overlap, whereas $\phi=0.80$ indicates relatively poor overlap; see Figure~\ref{fig:phi-overlap}. The correlation $\rho=\operatorname{cor}\{W,Y(0)\}$ measures the extent to which the covariate direction distinguishing the RCT and EC populations is also prognostic for control EFS. We use $\rho=0.20$ as the primary value and examine sensitivity to alternative values. {A limitation of this calculation is that it does not account for censoring, whereas the planned analysis uses survival pseudo-values. Consequently, the resulting $\nrct$ may underestimate the required sample size when censoring is substantial. In practice, the binary-based calculation may be inflated to account for the anticipated loss of information due to censoring; extending the proposed framework directly to censored time-to-event outcomes is left for future work.}

Under the primary assumptions, the smallest integer RCT sample size achieving 80\% power is $\appPrimaryMinimum$. A 3:1 allocation requires a complete four-person randomization block, so the operational recommendation is to randomize $\appPrimaryPlanned$ participants: $\appPrimaryTreatment$ to the experimental arm and $\appPrimaryControl$ to the concurrent control arm. The resulting asymptotic power is $\appPrimaryPower$.

The R code reproducing the primary analysis is provided in Appendix Section~\ref{sec:R}. Sensitivity analyses for the number of outcome-drift-free ECs, treatment allocation, overlap coefficient, and correlation coefficient are reported in Appendix Section~\ref{sec:application-sensitivity}.

\section{Discussion}

We developed a prospective sample-size framework for HCTs that supplements five conventional RCT design parameters with three scalar parameters characterizing EC comparability. Under the proposed working models, the number of outcome-drift-free ECs, the overlap coefficient, and the correlation coefficient determine the asymptotic variance of the IPW estimator and hence the required RCT sample size. This 5+3 design retains the familiar structure of conventional RCT planning while incorporating features specific to HCTs. Investigators can specify plausible values for the three HCT-specific parameters and assess their impact on sample size and power. This is important because a larger number of ECs does not necessarily provide greater information when outcome drift is present or covariate overlap is limited. The clinical application illustrates how assumptions about EC comparability translate into prospective design consequences.

The proposed calculations rely on working assumptions, including a logistic sampling model, a normal approximation for the linear sampling propensity score, a homoscedastic linear model for continuous outcomes or a logistic model for binary outcomes, and no conditional outcome drift among the ECs designated as outcome-drift-free. These outcome and sampling models are widely used in clinical trial planning, drug development, and regulatory applications. These assumptions cannot be fully assessed before recruitment, but prospective sample-size calculation necessarily requires assumptions about unobserved quantities. Our parameterization makes these assumptions explicit and interpretable, while simulations support the accuracy of the resulting calculations across a range of settings.

Several limitations of the framework point to concrete refinements. First, the design variance treats the sampling propensity score as known; accounting for its estimation could reduce the conservatism observed in our simulations, at the cost of additional design inputs. Second, the framework targets power under correctly specified working models but does not by itself protect type I error against residual outcome drift among the borrowed ECs; pairing the calculated design with a prespecified selective-borrowing analysis \citep{gao2025improving,zhu2025enhancing,liu2025robust} and reporting drift tipping points or other sensitivity analyses \citep{gordon2025non,liu2026value} would strengthen its credibility in regulatory settings. Third, the three HCT-specific parameters are elicited rather than estimated, and misspecification, particularly an optimistic overlap coefficient, can leave the trial underpowered; a hybrid strategy that begins with the proposed prospective calculation and prespecifies a blinded sample-size reassessment at an interim analysis \citep{guo2024adaptive,kojima2026sample} would allow early data to correct misspecified inputs while preserving the prospective character of the design. Fourth, we focus on the IPW estimator because its design-stage variance can be characterized using limited inputs; extending the framework to more efficient estimators \citep{li2023improving} may further reduce the required sample size. Fifth, the current framework does not directly accommodate censoring; extending the sample-size calculation to time-to-event outcomes \citep{gao2025doubly,yang2026sample} is an important direction for future work.

\vspace{-20pt}
\section*{Acknowledgment}
This project is supported by the Food and Drug Administration (FDA) of the U.S. Department of Health and Human Services (HHS) as part of a financial assistance award, U01FD007934, totaling \$2,556,429 over three years, funded by the FDA/HHS. This work is also supported by R01AG066883, funded by the NIH/HHS. The contents are those of the authors and do not necessarily represent the official views of, nor an endorsement by, FDA/HHS, NIH/HHS, or the U.S. Government.





\bibliography{ref.bib}
\newpage
\appendix

\section{Proofs}

\subsection{Proof of Theorem~\ref{thm:ipw_asymp}}
\begin{proof}
Let
\[
d(X)=1-\bar{\pi}_A\pi_S(X).
\]

\medskip
\noindent\textbf{Consistency.}
For the treated component, randomization and consistency imply
\[
\mathbb{E}(w_1)
=
\mathbb{E}\left(\frac{SA}{\bar{\pi}_A}\right)
=
\bar{\pi}_S
\]
and
\[
\begin{aligned}
\mathbb{E}(w_1Y)
&=
\mathbb{E}\left\{
\frac{SA}{\bar{\pi}_A}Y(1)
\right\}\\
&=
\mathbb{E}\{SY(1)\}
=
\bar{\pi}_S\theta_1.
\end{aligned}
\]

For the control component,
\[
\begin{aligned}
\mathbb{E}(w_0\mid X)
&=
\frac{\pi_S(X)}{d(X)}
\left[
\pi_S(X)(1-\bar{\pi}_A)+1-\pi_S(X)
\right]\\
&=
\pi_S(X),
\end{aligned}
\]
and hence
\[
\mathbb{E}(w_0)=\bar{\pi}_S.
\]
Similarly, consistency and conditional mean exchangeability give
\[
\begin{aligned}
\mathbb{E}(w_0Y\mid X)
&=
\frac{\pi_S(X)}{d(X)}
\left[
\pi_S(X)(1-\bar{\pi}_A)
\mathbb{E}\{Y(0)\mid X,S=1\}
\right.\\
&\qquad\left.
+
\{1-\pi_S(X)\}
\mathbb{E}\{Y(0)\mid X,S=0\}
\right]\\
&=
\pi_S(X)\mathbb{E}\{Y(0)\mid X,S=1\}.
\end{aligned}
\]
Therefore,
\[
\begin{aligned}
\mathbb{E}(w_0Y)
&=
\mathbb{E}\left[
\pi_S(X)\mathbb{E}\{Y(0)\mid X,S=1\}
\right]\\
&=
\mathbb{E}\{SY(0)\}
=
\bar{\pi}_S\theta_0.
\end{aligned}
\]

It follows from the law of large numbers and $\bar{\pi}_S>0$ that
\[
\widehat{\theta}_a
=
\frac{\mathbb{P}_n(w_aY)}{\mathbb{P}_n w_a}
\overset{p}{\longrightarrow}
\theta_a,
\qquad a=0,1,
\]
and consequently
\[
\widehat{\tau}\overset{p}{\longrightarrow}\tau.
\]

\medskip
\noindent\textbf{Asymptotic linear representation.}
For $a=0,1$,
\[
\widehat{\theta}_a-\theta_a
=
\frac{
\mathbb{P}_n\{w_a(Y-\theta_a)\}
}{
\mathbb{P}_n w_a
}.
\]
Since $\mathbb{P}_n w_a\overset{p}{\longrightarrow}\bar{\pi}_S$, Slutsky's theorem yields
\[
\sqrt{n}(\widehat{\theta}_a-\theta_a)
=
\frac{1}{\bar{\pi}_S}
\frac{1}{\sqrt{n}}
\sum_{i=1}^n
w_{a,i}(Y_i-\theta_a)
+o_p(1).
\]
Thus
\[
\sqrt{n}(\widehat{\tau}-\tau)
=
\frac{1}{\sqrt{n}}
\sum_{i=1}^n
\{\phi_1(O_i)-\phi_0(O_i)\}
+o_p(1),
\]
where
\[
\phi_a(O)
=
\frac{w_a(Y-\theta_a)}{\bar{\pi}_S},
\qquad a=0,1.
\]
The preceding calculations imply
\[
\mathbb{E}\{\phi_1(O)\}
=
\mathbb{E}\{\phi_0(O)\}
=
0.
\]

\medskip
\noindent\textbf{Asymptotic variance.}
For the treated component,
\[
\mathbb{V}\{\phi_1(O)\}
=
\frac{1}{\bar{\pi}_S^2\bar{\pi}_A^2}
\mathbb{E}\left[
SA\{Y(1)-\theta_1\}^2
\right].
\]
By treatment randomization,
\[
\begin{aligned}
\mathbb{E}\left[
SA\{Y(1)-\theta_1\}^2
\right]
&=
\bar{\pi}_A
\mathbb{E}\left[
S\{Y(1)-\theta_1\}^2
\right]\\
&=
\bar{\pi}_A\bar{\pi}_S
\mathbb{V}\{Y(1)\mid S=1\}.
\end{aligned}
\]
Therefore,
\[
V_1
=
\mathbb{V}\{\phi_1(O)\}
=
\frac{\mathbb{V}\{Y(1)\mid S=1\}}
{\bar{\pi}_A\bar{\pi}_S}.
\]

Conditional mean exchangeability and conditional second-moment
transportability imply
\[
\begin{aligned}
M_0(X)
&:=
\mathbb{E}\left[
\{Y(0)-\theta_0\}^2\mid X,S=1
\right]\\
&=
\mathbb{E}\left[
\{Y(0)-\theta_0\}^2\mid X,S=0
\right].
\end{aligned}
\]
Therefore,
\[
\begin{aligned}
&\mathbb{E}\left[
w_0^2\{Y-\theta_0\}^2\mid X
\right]\\
&\quad=
\frac{\pi_S(X)^2}{d(X)^2}
\left[
\pi_S(X)(1-\bar{\pi}_A)+1-\pi_S(X)
\right]M_0(X)\\
&\quad=
\frac{\pi_S(X)^2}{d(X)}M_0(X).
\end{aligned}
\]
Moreover,
\[
\mathbb{E}\left[
\{Y(0)-\theta_0\}^2\mid X
\right]
=
M_0(X),
\]
so iterated expectation gives
\[
\begin{aligned}
\mathbb{V}\{\phi_0(O)\}
&=
\frac{1}{\bar{\pi}_S^2}
\mathbb{E}\left[
\frac{\pi_S(X)^2}{d(X)}M_0(X)
\right]\\
&=
\frac{1}{\bar{\pi}_S^2}
\mathbb{E}\left[
\frac{\pi_S(X)^2}
{1-\bar{\pi}_A\pi_S(X)}
\{Y(0)-\theta_0\}^2
\right]\\
&=
V_0.
\end{aligned}
\]

\medskip
\noindent\textbf{Asymptotic distribution.} The stated moment conditions ensure that $V_1$ and $V_0$ are finite.
In particular, because
\[
d(X)\geq1-\bar{\pi}_A
\quad\text{and}\quad
\pi_S(X)^2\leq\pi_S(X),
\]
we have
\[
\begin{aligned}
\mathbb{E}\left[
\frac{\pi_S(X)^2}{d(X)}M_0(X)
\right]
&\leq
\frac{1}{1-\bar{\pi}_A}
\mathbb{E}\{\pi_S(X)M_0(X)\}\\
&=
\frac{\bar{\pi}_S}{1-\bar{\pi}_A}
\mathbb{E}\left[
\{Y(0)-\theta_0\}^2\mid S=1
\right]
<\infty.
\end{aligned}
\]

Finally,
\[
w_1w_0
=
\frac{SA}{\bar{\pi}_A}
\frac{\pi_S(X)(1-A)}{d(X)}
=
0
\qquad\text{almost surely}.
\]
Since both influence functions have mean zero,
\[
\operatorname{Cov}\{\phi_1(O),\phi_0(O)\}=0.
\]
Hence
\[
\mathbb{V}\{\phi_1(O)-\phi_0(O)\}
=
V_1+V_0.
\]
The central limit theorem and Slutsky's theorem now give
\[
\sqrt{n}(\widehat{\tau}-\tau)
\overset{d}{\longrightarrow}
N(0,V_1+V_0),
\]
which proves the result.
\end{proof}

\subsection{Proof of Theorem~\ref{thm:binary_variance}}
\begin{proof}
By the balancing property of the sampling propensity score,
$S\perp X\mid W$. Therefore, Assumption~\ref{ass:ec} implies
\[
\mathbb{E}\{Y(0)\mid S,W\}
=
\mathbb{E}\{Y(0)\mid W\}.
\]
Hence,
\[
\begin{aligned}
\theta_0
&=
\mathbb{E}\{Y(0)\mid S=1\}\\
&=
\frac{1}{\bar{\pi}_S}
\mathbb{E}\left[
S Y(0)
\right]\\
&=
\frac{1}{\bar{\pi}_S}
\mathbb{E}\left[
\mathbb{E}(S\mid W)
\mathbb{E}\{Y(0)\mid W\}
\right]\\
&=
\frac{1}{\bar{\pi}_S}
\mathbb{E}\left[
\expit(W)\expit(a+bW)
\right],
\end{aligned}
\]
which gives \eqref{eq:binary_mean}.

Moreover,
\[
\rho
=
\operatorname{cor}\{W,Y(0)\}
=
\frac{\operatorname{Cov}\{W,Y(0)\}}
{\left[
\mathbb{V}(W)\mathbb{V}\{Y(0)\}
\right]^{1/2}}.
\]
Because $Y(0)$ is binary and
\[
\mathbb{E}\{Y(0)\}
=
\mathbb{E}\left[
\mathbb{E}\{Y(0)\mid W\}
\right]
=
\mathbb{E}\{\expit(a+bW)\},
\]
we have
\[
\mathbb{V}\{Y(0)\}
=
\mathbb{E}\{\expit(a+bW)\}
\left[
1-\mathbb{E}\{\expit(a+bW)\}
\right].
\]
Furthermore, since $\mu_W=\mathbb{E}(W)$,
\[
\begin{aligned}
\operatorname{Cov}\{W,Y(0)\}
&=
\mathbb{E}\left[
(W-\mu_W)Y(0)
\right]\\
&=
\mathbb{E}\left[
(W-\mu_W)
\mathbb{E}\{Y(0)\mid W\}
\right]\\
&=
\mathbb{E}\left[
(W-\mu_W)\expit(a+bW)
\right].
\end{aligned}
\]
Combining these expressions with
$\mathbb{V}(W)=\sigma_W^2$ gives
\eqref{eq:binary_corr}.

Finally, because $Y(0)$ is binary, $Y(0)^2=Y(0)$. Therefore,
\[
\begin{aligned}
\mathbb{E}\left[
\{Y(0)-\theta_0\}^2\mid W
\right]
&=
\mathbb{E}\{Y(0)^2\mid W\}
-2\theta_0\mathbb{E}\{Y(0)\mid W\}
+\theta_0^2\\
&=
(1-2\theta_0)\expit(a+bW)+\theta_0^2.
\end{aligned}
\]
Substituting this expression and
$\pS(X)=\expit(W)$ into \eqref{eq:var}
yields \eqref{eq:V0_binary}.
\end{proof}

\subsection{Proof of Theorem~\ref{thm:continuous_variance}}
\begin{proof}
Under \eqref{eq:Ycon},
\[
\theta_0
=
\mathbb{E}\{Y(0)\mid S=1\}
=
a+b\mu_{W,1},
\]
and hence
\[
Y(0)-\theta_0
=
b(W-\mu_{W,1})+\varepsilon.
\]

Therefore,
\begin{equation}
\label{eq:proof_variance_decomp}
\sigma_Y^2
=
\mathbb{V}\{Y(0)\mid S=1\}
=
b^2\sigma_{W,1}^2+\sigma_\varepsilon^2,
\end{equation}
where the covariance term is zero because
$\mathbb{E}(\varepsilon\mid W,S)=0$, and
$\mathbb{V}(\varepsilon\mid S=1)=\sigma_\varepsilon^2$ follows from
$\mathbb{V}(\varepsilon\mid W,S)=\sigma_\varepsilon^2$.

Moreover,
\begin{equation}
\label{eq:proof_correlation}
\rho
=
\operatorname{cor}\{W,Y(0)\}
=
\frac{b\sigma_W}
{\{b^2\sigma_W^2+\sigma_\varepsilon^2\}^{1/2}},
\end{equation}
where the covariance term between $W$ and $\varepsilon$ is zero by
$\mathbb{E}(\varepsilon\mid W,S)=0$.

Combining \eqref{eq:proof_variance_decomp} and
\eqref{eq:proof_correlation} gives
\[
b^2
=
\frac{\rho^2\sigma_Y^2}
{\rho^2\sigma_{W,1}^2+(1-\rho^2)\sigma_W^2},
\]
and
\[
\sigma_\varepsilon^2
=
\frac{(1-\rho^2)\sigma_W^2\sigma_Y^2}
{\rho^2\sigma_{W,1}^2+(1-\rho^2)\sigma_W^2}.
\]
By the definition of $\kappa$,
\[
b^2=\frac{\rho^2\kappa\sigma_Y^2}{\sigma_W^2},
\qquad
\sigma_\varepsilon^2=(1-\rho^2)\kappa\sigma_Y^2.
\]

Thus,
\[
\begin{aligned}
\mathbb{E}\left[
\{Y(0)-\theta_0\}^2\mid W
\right]
&=
b^2(W-\mu_{W,1})^2+\sigma_\varepsilon^2\\
&=
\sigma_Y^2\kappa
\left\{
(1-\rho^2)
+
\rho^2
\frac{(W-\mu_{W,1})^2}{\sigma_W^2}
\right\}.
\end{aligned}
\]

Substituting this expression into \eqref{eq:var} with
$\pS(X)=\expit(W)$ gives \eqref{eq:V0_continuous}.
\end{proof}

\section{Additional Simulation Results}

\begin{figure}[t]
\centering
\includegraphics[width=\linewidth]{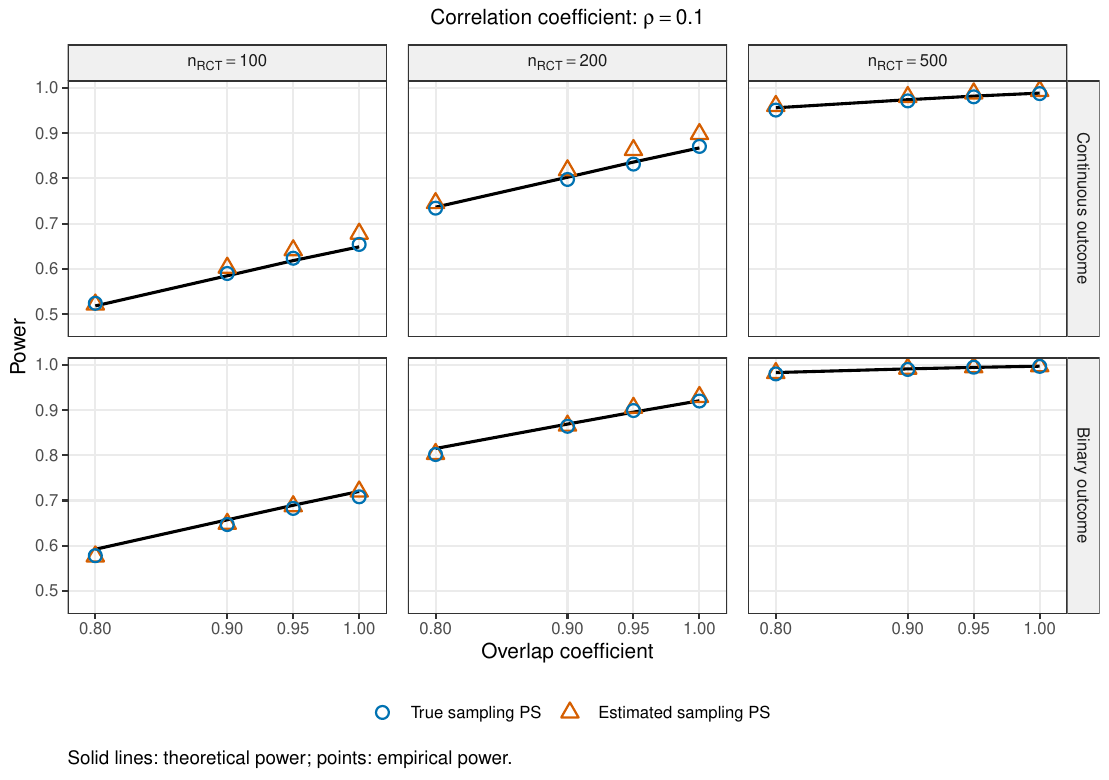}
\caption{Theoretical and empirical power for correlation coefficient
$\rho=0.1$, by outcome type, RCT sample size, and overlap coefficient. Solid
lines show the proposed theoretical power, and points show empirical power
from 10,000 replicates. The perfect-overlap endpoint has $\rho=0$ by
convention.}
\label{fig:simulation-power-rho01}
\end{figure}

Figure~\ref{fig:simulation-power-rho01} presents the power results for
$\rho=0.1$. The agreement between theoretical and empirical power and the
increase in power with the overlap coefficient are similar to the findings for
$\rho=0.3$ in Figure~\ref{fig:simulation-power-rho03}.

Table~\ref{tab:simulation-power} reports theoretical and
empirical power for the binary and continuous outcomes.

Tables~\ref{tab:simulation-validation-continuous}
and~\ref{tab:simulation-validation-binary} report the bias, empirical standard
deviation, mean estimated standard error, and coverage for every scenario.
The maximum absolute standardized bias, defined as absolute bias divided by the empirical standard deviation, was 3.3\%, supporting the consistency of the IPW estimator. With the true sampling propensity score, empirical-to-theoretical variance ratios ranged from 0.976 to 1.032 and coverage from 0.943 to 0.951, indicating accurate design calculations despite the nonnormality of $W$ induced by mixed continuous and binary covariates. With the estimated sampling propensity score, variance ratios were 0.909--0.995 for binary outcomes and 0.725--0.915 for continuous outcomes, reflecting efficiency gains from propensity score estimation. Accordingly, direct plug-in intervals were conservative, with coverage of 0.958--0.977.

\input{table/simulation_binary.tex}

\input{table/simulation_validation_binary.tex}

\input{table/simulation_validation_continuous.tex}

\section{R Code for the Primary Design Calculation}
\label{sec:R}

\begin{lstlisting}[language=R]
install.packages("remotes")  # only if needed
remotes::install_github("ke-zhu/hctdesign")
library(hctdesign)
fit <- hct_sample_size(
  target_power = 0.80,
  tau = 0.603 - 0.485,
  pi_a = 3 / 4,
  alpha = 0.10,
  n_ec = 300,
  phi = 0.90,
  rho = 0.20,
  outcome = "binary",
  theta0 = 0.485,
  alternative = "greater"
)
fit$n_rct                   # 173
ceiling(fit$n_rct / 4) * 4  # 176 for complete four-person randomization block
fit                         # full results
\end{lstlisting}

The complete code is available at \url{github.com/ke-zhu/hct-sample-size}.

\section{Sensitivity Analyses for the Real Data Application}
\label{sec:application-sensitivity}

\subsection{EC Sample Size and Treatment Allocation}

Table~\ref{tab:application-ec-allocation} jointly varies the number of outcome-drift-free ECs and the randomization ratio while holding all other inputs at their primary values. 

\input{table/application_ec_allocation.tex}

\subsection{Overlap and Correlation Coefficients}

Table~\ref{tab:application-overlap-correlation} examines uncertainty in
population comparability under the primary values $\nec=300$ and
$\bpA=0.75$. The operational RCT sample size increases from 140 under perfect
overlap to 156, 176, and 220 as $\phi$ decreases to 0.95, 0.90, and 0.80,
respectively. Thus, the overlap coefficient has a substantial impact and
should be justified using the expected covariate distributions. Over
the range considered, changing $\rho$ from 0 to 0.30 changes the mathematical
minimum by at most one participant and does not change the block-rounded
sample size. This small effect is specific to the present EFS assumptions and
should not be assumed in other applications.

\input{table/application_overlap_correlation.tex}

\end{document}

%% file: notation.tex
\usepackage{url}

\usepackage{tabularx}
\usepackage{authblk}
\usepackage{xcolor}

\newcommand{\T}{\top}

\newcommand{\nrct}{n_{\rm RCT}}
\newcommand{\nec}{n_{\rm EC}}

\newcommand{\bpA}{\bar{\pi}_A}
\newcommand{\pS}{\pi_S}
\newcommand{\bpS}{\bar{\pi}_S}

\newcommand{\expit}{\mathrm{expit}}
\newcommand{\logit}{\operatorname{logit}}

\usepackage{listings}

%% file: table/application_numbers.tex
\newcommand{\appPrimaryMinimum}{173}
\newcommand{\appPrimaryPlanned}{176}
\newcommand{\appPrimaryTreatment}{132}
\newcommand{\appPrimaryControl}{44}
\newcommand{\appPrimaryPower}{80.4\%}

%% file: table/application_parameters.tex
\begin{table}[p]
\centering
\caption{Design parameters for sample size calculation in the HCT application.}
\label{tab:application-parameters}
\small
\setlength{\tabcolsep}{6pt}
\begin{tabular}{
  >{\centering\arraybackslash}p{1.8cm}
  >{\centering\arraybackslash}p{1.8cm}
  >{\raggedright\arraybackslash}p{10.0cm}}
\toprule
Input & Primary value & Practical interpretation \\
\midrule
$\theta_0$ & 0.485
& Expected 48-month EFS rate under control in the RCT population \\
$\tau$ & 0.118
& Absolute treatment effect, corresponding to an increase in 48-month EFS from 0.485 to 0.603 \\
$\bpA$ & 0.75
& Treatment allocation probability, corresponding to 3:1 randomization \\
$\alpha$ & 0.10
& One-sided type I error for the superiority test \\
$1-\beta$ & 0.80
& Target power for detecting the prespecified treatment effect \\
\midrule
$n_{\rm EC}$ & 300
& Number of EC patients retained as outcome-drift-free after the prespecified compatibility assessment \\
$\phi$ & 0.90
& Degree of covariate overlap between the RCT and EC populations, measured through their sampling propensity score distributions \\
$\rho$ & 0.20
& Association between the covariate direction distinguishing the RCT and EC populations and the control outcome, quantified by $\operatorname{cor}\{W,Y(0)\}$ \\
\bottomrule
\end{tabular}
\end{table}

%% file: table/simulation_binary.tex
\begin{table}[t]
\centering
\caption{Theoretical and empirical power for binary and continuous outcomes.}
\label{tab:simulation-power}
\small
\setlength{\tabcolsep}{4pt}
\begin{tabular}{rrrcccccc}
\toprule
& & &
\multicolumn{3}{c}{Binary outcome} &
\multicolumn{3}{c}{Continuous outcome} \\
\cmidrule(lr){4-6} \cmidrule(lr){7-9}
$n_{\mathrm{RCT}}$ & $\phi$ & $\rho$
& Theoretical & True PS & Estimated PS
& Theoretical & True PS & Estimated PS \\
\midrule
100 & 0.80 & 0.1 & 0.592 & 0.578 & 0.577 & 0.518 & 0.524 & 0.521 \\
100 & 0.80 & 0.3 & 0.570 & 0.556 & 0.556 & 0.515 & 0.524 & 0.522 \\
100 & 0.90 & 0.1 & 0.657 & 0.647 & 0.648 & 0.585 & 0.590 & 0.603 \\
100 & 0.90 & 0.3 & 0.642 & 0.628 & 0.634 & 0.582 & 0.592 & 0.601 \\
100 & 0.95 & 0.1 & 0.689 & 0.682 & 0.687 & 0.618 & 0.624 & 0.641 \\
100 & 0.95 & 0.3 & 0.680 & 0.675 & 0.679 & 0.616 & 0.622 & 0.637 \\
100 & 1.00 & 0.0 & 0.720 & 0.708 & 0.720 & 0.649 & 0.654 & 0.678 \\
\midrule
200 & 0.80 & 0.1 & 0.815 & 0.802 & 0.803 & 0.736 & 0.734 & 0.745 \\
200 & 0.80 & 0.3 & 0.794 & 0.782 & 0.776 & 0.734 & 0.727 & 0.742 \\
200 & 0.90 & 0.1 & 0.869 & 0.864 & 0.866 & 0.803 & 0.798 & 0.818 \\
200 & 0.90 & 0.3 & 0.855 & 0.849 & 0.850 & 0.800 & 0.804 & 0.818 \\
200 & 0.95 & 0.1 & 0.895 & 0.899 & 0.905 & 0.836 & 0.832 & 0.863 \\
200 & 0.95 & 0.3 & 0.887 & 0.889 & 0.892 & 0.834 & 0.827 & 0.855 \\
200 & 1.00 & 0.0 & 0.921 & 0.920 & 0.929 & 0.867 & 0.871 & 0.898 \\
\midrule
500 & 0.80 & 0.1 & 0.983 & 0.980 & 0.982 & 0.956 & 0.951 & 0.960 \\
500 & 0.80 & 0.3 & 0.979 & 0.974 & 0.978 & 0.956 & 0.951 & 0.959 \\
500 & 0.90 & 0.1 & 0.991 & 0.990 & 0.991 & 0.974 & 0.971 & 0.980 \\
500 & 0.90 & 0.3 & 0.989 & 0.988 & 0.987 & 0.974 & 0.973 & 0.979 \\
500 & 0.95 & 0.1 & 0.994 & 0.995 & 0.995 & 0.982 & 0.980 & 0.988 \\
500 & 0.95 & 0.3 & 0.993 & 0.993 & 0.994 & 0.982 & 0.979 & 0.988 \\
500 & 1.00 & 0.0 & 0.997 & 0.997 & 0.997 & 0.989 & 0.987 & 0.993 \\
\bottomrule
\end{tabular}

\begin{minipage}{0.96\linewidth}
\vspace{5pt}
\footnotesize
Each setting used 10,000 replicates. At perfect overlap ($\phi=1$), $\rho=0$ by convention. PS denotes sampling propensity score.
\end{minipage}
\end{table}

%% file: table/simulation_validation_binary.tex
\begin{table}[t]
\centering
\caption{Operating characteristics of the IPW estimator for binary outcomes.}
\label{tab:simulation-validation-binary}
\footnotesize
\setlength{\tabcolsep}{2.5pt}
\begin{tabular}{rrrrrrrrrrr}
\toprule
 & & & \multicolumn{4}{c}{True sampling PS} & \multicolumn{4}{c}{Estimated sampling PS} \\
\cmidrule(lr){4-7} \cmidrule(lr){8-11}
$n_{\mathrm{RCT}}$ & $\phi$ & $\rho$ & Bias & Emp. SD & Mean SE & Coverage & Bias & Emp. SD & Mean SE & Coverage \\
\midrule
100 & 0.80 & 0.1 & -0.0016 & 0.0739 & 0.0731 & 0.944 & -0.0015 & 0.0721 & 0.0736 & 0.952 \\
100 & 0.80 & 0.3 & -0.0007 & 0.0747 & 0.0749 & 0.948 & -0.0007 & 0.0734 & 0.0754 & 0.956 \\
100 & 0.90 & 0.1 & -0.0016 & 0.0676 & 0.0675 & 0.948 & -0.0017 & 0.0656 & 0.0679 & 0.957 \\
100 & 0.90 & 0.3 & -0.0020 & 0.0690 & 0.0686 & 0.947 & -0.0019 & 0.0667 & 0.0691 & 0.957 \\
100 & 0.95 & 0.1 & -0.0015 & 0.0652 & 0.0649 & 0.945 & -0.0015 & 0.0629 & 0.0653 & 0.955 \\
100 & 0.95 & 0.3 & -0.0018 & 0.0651 & 0.0656 & 0.948 & -0.0018 & 0.0630 & 0.0660 & 0.958 \\
100 & 1.00 & 0.0 & -0.0021 & 0.0627 & 0.0625 & 0.943 & -0.0019 & 0.0600 & 0.0628 & 0.957 \\
\midrule
200 & 0.80 & 0.1 & 0.0001 & 0.0568 & 0.0561 & 0.946 & -0.0001 & 0.0558 & 0.0564 & 0.950 \\
200 & 0.80 & 0.3 & -0.0005 & 0.0579 & 0.0575 & 0.946 & -0.0007 & 0.0572 & 0.0577 & 0.950 \\
200 & 0.90 & 0.1 & 0.0000 & 0.0522 & 0.0519 & 0.947 & 0.0000 & 0.0508 & 0.0521 & 0.954 \\
200 & 0.90 & 0.3 & 0.0005 & 0.0536 & 0.0529 & 0.946 & 0.0003 & 0.0526 & 0.0532 & 0.951 \\
200 & 0.95 & 0.1 & 0.0006 & 0.0495 & 0.0497 & 0.949 & 0.0003 & 0.0477 & 0.0499 & 0.958 \\
200 & 0.95 & 0.3 & -0.0001 & 0.0501 & 0.0504 & 0.950 & -0.0004 & 0.0488 & 0.0506 & 0.956 \\
200 & 1.00 & 0.0 & 0.0000 & 0.0481 & 0.0474 & 0.946 & 0.0000 & 0.0464 & 0.0476 & 0.954 \\
\midrule
500 & 0.80 & 0.1 & 0.0005 & 0.0397 & 0.0392 & 0.946 & 0.0005 & 0.0391 & 0.0393 & 0.951 \\
500 & 0.80 & 0.3 & 0.0003 & 0.0404 & 0.0401 & 0.946 & 0.0003 & 0.0400 & 0.0401 & 0.950 \\
500 & 0.90 & 0.1 & 0.0004 & 0.0372 & 0.0370 & 0.947 & 0.0002 & 0.0364 & 0.0371 & 0.953 \\
500 & 0.90 & 0.3 & 0.0005 & 0.0382 & 0.0377 & 0.946 & 0.0004 & 0.0377 & 0.0378 & 0.949 \\
500 & 0.95 & 0.1 & 0.0005 & 0.0360 & 0.0356 & 0.944 & 0.0004 & 0.0351 & 0.0357 & 0.953 \\
500 & 0.95 & 0.3 & 0.0004 & 0.0363 & 0.0362 & 0.950 & 0.0002 & 0.0356 & 0.0363 & 0.953 \\
500 & 1.00 & 0.0 & 0.0002 & 0.0344 & 0.0340 & 0.947 & 0.0002 & 0.0333 & 0.0341 & 0.957 \\
\bottomrule
\end{tabular}
\begin{minipage}{0.9\linewidth}
\vspace{5pt}
\footnotesize Bias is the Monte Carlo mean of the estimated treatment effect minus the true effect; Emp. SD is its empirical standard deviation; Mean SE is the mean estimated standard error; Coverage is the empirical coverage of the nominal 95\% confidence interval. Each setting used 10,000 replicates. PS denotes sampling propensity score.
\end{minipage}
\end{table}

%% file: table/simulation_validation_continuous.tex
\begin{table}[t]
\centering
\caption{Operating characteristics of the IPW estimator for continuous outcomes.}
\label{tab:simulation-validation-continuous}
\footnotesize
\setlength{\tabcolsep}{2.5pt}
\begin{tabular}{rrrrrrrrrrr}
\toprule
 & & & \multicolumn{4}{c}{True sampling PS} & \multicolumn{4}{c}{Estimated sampling PS} \\
\cmidrule(lr){4-7} \cmidrule(lr){8-11}
$n_{\mathrm{RCT}}$ & $\phi$ & $\rho$ & Bias & Emp. SD & Mean SE & Coverage & Bias & Emp. SD & Mean SE & Coverage \\
\midrule
100 & 0.80 & 0.1 & 0.0011 & 0.1513 & 0.1495 & 0.943 & 0.0016 & 0.1377 & 0.1507 & 0.966 \\
100 & 0.80 & 0.3 & 0.0014 & 0.1511 & 0.1505 & 0.944 & 0.0023 & 0.1380 & 0.1515 & 0.967 \\
100 & 0.90 & 0.1 & -0.0005 & 0.1383 & 0.1374 & 0.946 & 0.0004 & 0.1222 & 0.1385 & 0.971 \\
100 & 0.90 & 0.3 & 0.0022 & 0.1385 & 0.1381 & 0.951 & 0.0025 & 0.1228 & 0.1390 & 0.971 \\
100 & 0.95 & 0.1 & -0.0006 & 0.1318 & 0.1319 & 0.949 & 0.0008 & 0.1152 & 0.1328 & 0.972 \\
100 & 0.95 & 0.3 & -0.0003 & 0.1325 & 0.1323 & 0.947 & 0.0001 & 0.1165 & 0.1332 & 0.972 \\
100 & 1.00 & 0.0 & -0.0001 & 0.1269 & 0.1272 & 0.948 & 0.0000 & 0.1091 & 0.1279 & 0.977 \\
\midrule
200 & 0.80 & 0.1 & 0.0007 & 0.1166 & 0.1156 & 0.948 & 0.0006 & 0.1078 & 0.1161 & 0.963 \\
200 & 0.80 & 0.3 & 0.0009 & 0.1170 & 0.1159 & 0.948 & 0.0008 & 0.1090 & 0.1165 & 0.964 \\
200 & 0.90 & 0.1 & -0.0004 & 0.1071 & 0.1065 & 0.948 & -0.0003 & 0.0969 & 0.1071 & 0.970 \\
200 & 0.90 & 0.3 & 0.0018 & 0.1075 & 0.1068 & 0.947 & 0.0009 & 0.0979 & 0.1074 & 0.970 \\
200 & 0.95 & 0.1 & -0.0009 & 0.1032 & 0.1018 & 0.946 & -0.0002 & 0.0908 & 0.1023 & 0.974 \\
200 & 0.95 & 0.3 & -0.0004 & 0.1040 & 0.1020 & 0.947 & 0.0000 & 0.0916 & 0.1026 & 0.973 \\
200 & 1.00 & 0.0 & -0.0002 & 0.0976 & 0.0972 & 0.950 & -0.0011 & 0.0850 & 0.0977 & 0.975 \\
\midrule
500 & 0.80 & 0.1 & -0.0005 & 0.0822 & 0.0818 & 0.946 & -0.0002 & 0.0783 & 0.0819 & 0.959 \\
500 & 0.80 & 0.3 & -0.0008 & 0.0821 & 0.0818 & 0.947 & -0.0006 & 0.0783 & 0.0819 & 0.958 \\
500 & 0.90 & 0.1 & -0.0009 & 0.0773 & 0.0768 & 0.945 & -0.0003 & 0.0721 & 0.0770 & 0.962 \\
500 & 0.90 & 0.3 & -0.0003 & 0.0769 & 0.0769 & 0.950 & -0.0004 & 0.0717 & 0.0771 & 0.964 \\
500 & 0.95 & 0.1 & -0.0014 & 0.0742 & 0.0740 & 0.947 & -0.0007 & 0.0678 & 0.0742 & 0.968 \\
500 & 0.95 & 0.3 & -0.0004 & 0.0747 & 0.0741 & 0.948 & -0.0003 & 0.0683 & 0.0743 & 0.966 \\
500 & 1.00 & 0.0 & -0.0005 & 0.0712 & 0.0708 & 0.948 & -0.0005 & 0.0642 & 0.0710 & 0.969 \\
\bottomrule
\end{tabular}
\begin{minipage}{0.9\linewidth}
\vspace{5pt}
\footnotesize Bias is the Monte Carlo mean of the estimated treatment effect minus the true effect; Emp. SD is its empirical standard deviation; Mean SE is the mean estimated standard error; Coverage is the empirical coverage of the nominal 95\% confidence interval. Each setting used 10,000 replicates. PS denotes sampling propensity score.
\end{minipage}
\end{table}

%% file: table/application_ec_allocation.tex
\begin{table}[htbp]
\centering
\caption{Required randomized sample size by number of outcome-drift-free ECs and treatment allocation.}
\label{tab:application-ec-allocation}
\setlength{\tabcolsep}{12pt}
\begin{tabular}{lrrr}
\toprule
$n_{\rm EC}$ & $1{:}1$ & $2{:}1$ & $3{:}1$ \\
\midrule
0 (RCT benchmark) & 318 & 360 & 428 \\
100 & 260 & 252 & 268 \\
200 & 232 & 207 & 204 \\
300 & 216 & 183 & \textbf{176} \\
\bottomrule
\end{tabular}
\par\smallskip
\begin{minipage}{0.88\linewidth}
\footnotesize Entries are RCT sample sizes rounded up to a complete randomization block. For $n_{\rm EC}>0$, $\phi=0.90$ and $\rho=0.20$; for the RCT benchmark, $\phi$ and $\rho$ are not applicable.
\end{minipage}
\end{table}

%% file: table/application_overlap_correlation.tex
\begin{table}[htbp]
\centering
\caption{Required randomized sample size by overlap and correlation coefficients.}
\label{tab:application-overlap-correlation}
\begin{tabular}{rrrrr}
\toprule
$\phi$ & $\rho=0$ & $\rho=0.10$ & $\rho=0.20$ & $\rho=0.30$ \\
\midrule
0.80 & 217 (220) & 217 (220) & 218 (220) & 219 (220) \\
0.90 & 173 (176) & 173 (176) & \textbf{173 (176)} & 174 (176) \\
0.95 & 153 (156) & 153 (156) & 153 (156) & 153 (156) \\
1.00 & 137 (140) & -- & -- & -- \\
\bottomrule
\end{tabular}
\par\smallskip
\begin{minipage}{0.88\linewidth}
\footnotesize Entries are minimum RCT sample sizes, with values rounded up to a complete four-person randomization block in parentheses. Calculations use $n_{\rm EC}=300$ and 3:1 allocation. At perfect overlap ($\phi=1$), $\rho=0$ by convention; the remaining cells are therefore not applicable.
\end{minipage}
\end{table}